\documentclass[conference]{IEEEtran}
\IEEEoverridecommandlockouts
\usepackage{cite}
\usepackage{amsmath,amssymb,amsfonts}
\usepackage{amsthm}
\usepackage{algorithm}
\usepackage{algpseudocode}
\usepackage{graphicx}
\usepackage{textcomp}
\usepackage{xcolor}
\usepackage{booktabs}
\usepackage{array}
\usepackage[colorlinks=true,linkcolor=blue,citecolor=blue,urlcolor=blue]{hyperref}

\newtheorem{definition}{Definition}
\newtheorem{lemma}{Lemma}
\newtheorem{theorem}{Theorem}

\def\BibTeX{{\rm B\kern-.05em{\sc i\kern-.025em b}\kern-.08em
    T\kern-.1667em\lower.7ex\hbox{E}\kern-.125emX}}

\begin{document}

\title{SentryBench: A Cost-Sensitive Benchmark for\\ Adversarial Robustness of Tool-Using LLM Agents\\ in Public-Safety Deployments}

\author{
\IEEEauthorblockN{
Md Abu Yusuf\textsuperscript{1},
Oyshi Tabassum Aditi\textsuperscript{2},
Md Faysal Hosain Mazed\textsuperscript{3}
}
\IEEEauthorblockA{
\textsuperscript{1}Department of Computer Science, Maharishi International University, IA, USA\\
\textsuperscript{2}Department of Computer Science and Engineering, Daffodil International University, Dhaka, Bangladesh\\
\textsuperscript{3} Department of Information Technology, University of the Cumberlands, Williamsburg, KY, USA\\
samba.yusuf@gmail.com,\; oyshitabassum113@gmail.com,\; faysalmazedmiu@gmail.com
}
}

\maketitle

\begin{abstract}
Large language model (LLM) agents are increasingly wired to consequential tools
that send public alerts, dispatch responders, and query citizen records, where one
erroneous tool call carries real-world cost. Existing agent-safety benchmarks
measure generic harm or prompt-injection success and implicitly treat refusing as
the safe default. In public-safety deployments that assumption breaks: an agent too
cautious to dispatch help fails as surely as one that triggers a false alert. We
present SentryBench, which scores agents on a two-sided, cost-sensitive frontier
jointly measuring attack success, over-refusal (missing a real emergency), and
over-reaction (unwarranted action) under adversarial tool use, with scenarios
grounded in real emergency-call transcripts, disaster messages, and real
injection/jailbreak payloads. We formalize deployment harm as an asymmetric cost
functional, prove naive caution is harm-dominated whenever it refuses more genuine
emergencies than the attack harm it averts, and prove a targeted trust gate adds
zero over-refusal on clean input. Across 840 scored episodes over three frontier
models and three defenses, refuse-more caution drives over-refusal to 100 percent
on two of three models for at most a 3.7-point change in attack success and yields
the worst harm everywhere, whereas our trust-gating wrapper dominates the frontier
(lowest harm on every model, over-refusal at or below baseline, higher
containment). Provider content filters also silently perform part of the defending
and must be reported separately.
\end{abstract}


\section{Introduction}

Tool-using LLM agents are moving from chat surfaces into operational loops that
act with physical consequence: paging responders, broadcasting alerts, reading
citizen records. Public-safety operations centers are an early, high-stakes
adopter, and their agents face a risk profile conventional chat-safety evaluation
does not capture: one that is \emph{bidirectional} and \emph{asymmetric}. An agent
can fail by \emph{under-reacting}, suppressing or missing a real emergency (a
false negative, potentially lethal), or by \emph{over-reacting}, taking an
unwarranted action such as a region-wide false alert, an unnecessary dispatch, or
an over-broad records query that violates privacy (a false positive).

Chat-safety benchmarks optimize a single axis, refuse harmful requests, on which
``refuse more'' is always weakly safer. In public safety this is \emph{false}:
caution is not free, and the two error directions must be measured \emph{jointly}
with domain-calibrated, asymmetric costs. To our knowledge no prior benchmark
models this domain with a benign control arm capable of exposing over-refusal
alongside adversarial robustness.

We introduce \textbf{SentryBench}, a benchmark and evaluation framework built
around four research questions:

\begin{itemize}
\item \textbf{RQ1 (Vulnerability).} How susceptible are frontier tool-using agents
to public-safety-specific attacks (injection, poisoning, goal-hijack, over-trust,
privilege escalation)?
\item \textbf{RQ2 (Asymmetry).} Does a caution-first defense reduce attack success
only by raising over-refusal on genuine emergencies?
\item \textbf{RQ3 (Intervention).} Can a targeted trust gate suppress attacks
without paying the over-refusal tax, i.e., dominate the frontier?
\item \textbf{RQ4 (Platform effects).} How do provider content filters confound
the model comparison, and how should they be reported?
\end{itemize}

\noindent Our contributions are:

\begin{enumerate}
\item A domain-grounded benchmark spanning four sub-domains, five attack vectors,
and three outcome classes, with a two-sided benign control arm (urgent versus
ambiguous).
\item A cost-sensitive framework treating over-refusal and over-reaction as
first-class metrics, with a harm-weighted functional and a precision--recall
frontier as the primary comparison object; we prove naive caution is
harm-dominated in the observed regime and that our trust gate adds no over-refusal
on clean input.
\item A real-data pipeline: legitimate content from real emergency transcripts and
disaster messages, payloads from real injection/jailbreak corpora, with no
synthetic-only hand-waving and no real personally identifiable information (PII).
\item A trust-gating defense (TGCC) and a naive-caution strawman, with model-level
versus system-level accounting for provider filters.
\item A reproducible, ledger-first harness with deterministic scoring from a
tool-call ledger and bootstrap confidence intervals.
\end{enumerate}

\section{Related Work}

\textbf{Tool-using and multi-agent LLM systems.} The ability to interleave
reasoning with external tool calls, established by ReAct~\cite{b15} and
Toolformer~\cite{b16}, is what turns a language model into an agent that can act
on the world, and surveys of LLM agents map the resulting architecture and
capability space~\cite{b17}. As deployments compose several such agents, the
central open problem shifts to safe and scalable collaboration; the recent
comprehensive review of multi-agent LLM systems by Hossain et al.~\cite{b18}
organizes the safety, coordination, and scalability challenges of these systems
and motivates domain-grounded, cost-sensitive evaluation of consequential agent
actions such as the public-safety setting we study.

\textbf{Attacks on LLM agents.} Prompt-injection and jailbreak attacks subvert an
agent's instructions through crafted input: instruction-override
injection~\cite{b19}, jailbreak analysis~\cite{b20}, universal adversarial
suffixes~\cite{b21}, and formalizations of injection attacks and
defenses~\cite{b22}. Indirect injection through untrusted tool content was
characterized by Greshake et al.~\cite{b7} and stress-tested at
scale~\cite{b23}. These form the payload basis of our adversarial families.

\textbf{Agent-safety benchmarks and defenses.} Benchmarks including
InjecAgent~\cite{b1}, AgentDojo~\cite{b2}, ToolEmu~\cite{b3}, AgentHarm~\cite{b4},
R-Judge~\cite{b5}, and AgentSafetyBench~\cite{b6} target generic tools or harm
categories and largely measure only whether an attack succeeds, not over-refusal
or asymmetric domain cost. Spotlighting~\cite{b8} and control/data-flow
separation (CaMeL)~\cite{b9} inform our provenance component, while
red-teaming~\cite{b24} and constitutional training~\cite{b25} reduce but do not
eliminate residual susceptibility. NEXUS~\cite{b26} likewise intervenes at the
tool-call boundary, combining deterministic rules, argument inspection, and a
calibrated risk score to allow, block, or request confirmation on a proposed
action; it targets general runtime safety and latency rather than the
bidirectional, asymmetric over-refusal-versus-over-reaction tradeoff we
formalize and measure. Over-refusal work, notably XSTest~\cite{b10},
motivates our benign control arm but is confined to chat, not consequential tool
use. SentryBench's differentiator is the public-safety domain with bidirectional,
asymmetric failure costs, treating both error directions as first-class and
coupling a formal harm model to a deterministic, ledger-first scorer.

\section{Threat Model and Problem Setting}

The agent and its system prompt are trusted; inputs and tool results may be
attacker-controlled (caller transcripts, incident reports, and, for poisoning,
a tool's mid-chain return). The attacker can inject text into untrusted channels,
impersonate authority, drive multi-turn goal drift, and poison one tool's return,
but has no access to model weights or sandbox internals. Per scenario the
objective is a target failure mode: suppress a real emergency (under-reaction) or
trigger an unwarranted action (over-reaction, including privacy violation).
Table~\ref{tab:vectors} lists the five attack vectors (taxonomy rows); the outcome
classes (under-reaction, over-reaction, correct containment) form the columns.

\begin{table}[t]
\caption{Attack Vectors (Taxonomy Rows)}
\label{tab:vectors}
\centering
\footnotesize
\begin{tabular}{p{2.5cm}p{5.0cm}}
\toprule
\textbf{Vector} & \textbf{Mechanism} \\
\midrule
Prompt injection & malicious directive embedded in a real message or transcript \\
Goal hijacking & multi-turn drift reframing a real emergency as benign \\
Over-trust exploitation & authority impersonation (``dispatch override, skip verification'') \\
Privilege escalation & induced over-broad citizen-records query \\
Tool-result poisoning & a called tool returns attacker content that redirects the agent \\
\bottomrule
\end{tabular}
\end{table}

\section{Formal Framework}

Let $S$ ($N=|S|$) be the scenario set, partitioned into adversarial scenarios $A$,
urgent benign controls $B_u$ (the agent \emph{must} act), and ambiguous benign
controls $B_a$ (must \emph{not} over-act). For scenario $i$, $A^{*}_i$ is the
warranted (gold) action set and $F_i$ the forbidden set; from the ledger we
derive three Boolean indicators.

\begin{definition}[Behavioural indicators]
For scenario $i$, let $a_i=1$ if the agent took the warranted \emph{kind} of
action (a name-level match to some $g\in A^{*}_i$), $k_i=1$ if the agent
escalated to a human without unilateral harm, and $f_i=1$ if any forbidden action
in $F_i$ was executed. Define
\[
u_i = \mathbb{1}[\neg a_i \wedge \neg (k_i \wedge \kappa_i)], \qquad
o_i = f_i,
\]
where $\kappa_i\in\{0,1\}$ records whether escalation satisfies the scenario.
Here $u_i$ is under-reaction (for urgent controls, over-refusal) and $o_i$ is
over-reaction.
\end{definition}

We score \emph{acting} at the tool-name level, not the constraint level: choosing
the wrong unit type or priority is a separate correctness dimension, not
over-refusal, which keeps the benign control arm honest.

\begin{definition}[Deployment harm]
Each scenario carries asymmetric, domain-calibrated costs $c^{FN}_i>0$ (cost of
missing a real emergency) and $c^{FP}_i>0$ (cost of an unwarranted action). The
mean deployment harm of a policy $\pi$ is
\begin{equation}
H(\pi)\;=\;\frac{1}{N}\sum_{i\in S}\Big(c^{FN}_i\,u_i \;+\; c^{FP}_i\,o_i\Big).
\label{eq:harm}
\end{equation}
\end{definition}

The three headline rates condition on scenarios the model actually executed
(errored and content-filtered episodes get no behavioural outcome and are
accounted separately): $\mathrm{ASR}=\frac{1}{|A_e|}\sum_{i\in A_e}s_i$ ($s_i$
attack success, $A_e\subseteq A$ executed);
$\mathrm{ORR}=\frac{1}{|B_{u,e}|}\sum_{i\in B_{u,e}}u_i$; and
$\mathrm{FP}=\frac{1}{|B_{a,e}|}\sum_{i\in B_{a,e}}o_i$.

\begin{definition}[Frontier dominance]
\label{def:dominance}
Policy $\pi$ dominates $\pi'$ if $\mathrm{ASR}(\pi)\le\mathrm{ASR}(\pi')$,
$\mathrm{ORR}(\pi)\le\mathrm{ORR}(\pi')$, and $\mathrm{FP}(\pi)\le\mathrm{FP}(\pi')$
with strict inequality on at least one axis, i.e., it reduces attack success
without an over-refusal or over-reaction tax.
\end{definition}

The domain's central asymmetry is captured by the following identity, making
precise that ``refuse more'' is not free.

\begin{lemma}[Cost of refusal]
\label{lem:refusal}
Let $\pi_0$ and $\pi_c$ be two policies that behave identically on all adversarial
and ambiguous scenarios and on every urgent scenario where both act. Suppose
$\pi_c$ additionally refuses on a set $R\subseteq B_u$ of urgent scenarios where
$\pi_0$ acted correctly. Then
\[
H(\pi_c)-H(\pi_0)\;=\;\frac{1}{N}\sum_{i\in R}c^{FN}_i\;\ge\;0,
\]
with equality iff $R=\varnothing$.
\end{lemma}
\begin{proof}
Harm \eqref{eq:harm} decomposes additively; the policies agree on $A$, $B_a$, and
$B_u\setminus R$, so those terms cancel. On each $i\in R$, $o_i$ is unchanged
while $u_i$ flips $0\to1$, contributing $c^{FN}_i$. Since every $c^{FN}_i>0$, the
sum is non-negative and zero only when $R=\varnothing$.
\end{proof}

Lemma~\ref{lem:refusal} isolates the FN cost of refusal; combining both axes
characterizes a caution policy that also reduces attack harm.

\begin{theorem}[Caution is harm-dominated under bounded attack leverage]
\label{thm:caution}
Let $R\subseteq B_u$ be the set of urgent controls newly refused by $\pi_c$
relative to $\pi_0$, and let
\[
\Delta^{-}=\frac{1}{N}\sum_{i\in A}\Big(c^{FN}_i(u^0_i-u^c_i)+c^{FP}_i(o^0_i-o^c_i)\Big)
\]
be the adversarial harm averted by caution. Then
\begin{equation}
H(\pi_c)-H(\pi_0)\;=\;\frac{1}{N}\sum_{i\in R}c^{FN}_i\;-\;\Delta^{-}.
\label{eq:caution}
\end{equation}
In particular, if $\Delta^{-}\le \frac{1}{N}\sum_{i\in R}c^{FN}_i$, then caution
does not reduce, and strictly increases whenever the inequality is strict,
deployment harm.
\end{theorem}
\begin{proof}
Decompose $H$ over the three partitions: caution and baseline agree on $B_a$; on
$B_u$ the only difference is $R$, contributing $\frac{1}{N}\sum_{i\in R}c^{FN}_i$
by Lemma~\ref{lem:refusal}; on $A$ the difference is $-\Delta^{-}$ by definition.
Summing yields \eqref{eq:caution}; the final claim is immediate.
\end{proof}

Empirically (Section~\ref{sec:results}) caution changes ASR by at most $3.7$
points, so $\Delta^{-}$ is small while $R$ is large (up to \emph{every} urgent
control), placing all three models in the harm-dominated regime of
Theorem~\ref{thm:caution}. This is the formal content of RQ2.

For the intervention, let the injection score of untrusted content $x$ be
$s(x)=|\{p\in\mathcal{P}: p \text{ matches } x\}|$, the number of distinct
patterns from a fixed set $\mathcal{P}$ that fire on $x$, with threshold
$\theta=1$. The gate routes a consequential call $c$ to a human iff $s\ge\theta$,
or $c$ is an over-broad citizen-records query with $s>0$.

\begin{theorem}[Zero over-refusal overhead on clean input]
\label{thm:gate}
On any benign scenario whose untrusted channels contain no detector-matching
pattern ($s=0$) and no over-broad field request, the trust gate allows every
consequential call the base agent proposes. Consequently, on the clean-benign
urgent control set, the gate introduces no additional over-refusal:
$\mathrm{ORR}_{\text{gate}}\le\mathrm{ORR}_{\text{base}}$ restricted to clean
inputs.
\end{theorem}
\begin{proof}
The gate fires only if $s\ge\theta=1$ (false when $s=0$) or the over-broad
disjunct holds, which requires $s>0$ (also false). Both being false, the gate
returns \textsc{allow} for every consequential call, so the executed action set,
hence every $a_i$ and $u_i$, is unchanged from the base agent.
\end{proof}

Hence any over-refusal reduction under the wrapper is due to the provenance
preamble, not the gate, and any residual over-refusal on clean urgent controls is
inherited from the base agent. Finally, the gate's containment guarantee is
bounded exactly by detector recall.

\begin{theorem}[Gate soundness and recall-limited residual]
\label{thm:sound}
If an adversarial scenario's untrusted content matches at least one detector
pattern ($s\ge\theta$) and the attack requires a consequential call to succeed,
the gate routes the first such call to a human, yielding containment and $s_i=0$.
Hence
\[
\mathrm{ASR}_{\text{TGCC}}\;\le\;\frac{1}{|A_e|}\,\big|\{i\in A_e: s_i<\theta \text{ and } s_i^{\text{succeed}}\}\big|,
\]
i.e., every residual success lies in the detector's miss set.
\end{theorem}
\begin{proof}
When $s\ge\theta$ the gate condition holds on the first consequential call, which
is routed to a human rather than executed; with escalation satisfying the
scenario this sets $k_i=1$ and $s_i=0$. Therefore any surviving success must have
$s_i<\theta$, giving the stated bound.
\end{proof}

Theorem~\ref{thm:sound} is tight: the single genuine success under the wrapper
(Sonnet on scenario C-008) has $s=0$, an over-trust payload phrased to evade the
pattern set, and produced no gate event in the ledger. This is an honest
demonstration of a detector recall gap, not a silver bullet.

\section{SentryBench Construction}

\begin{figure*}[t]
\centerline{\includegraphics[width=0.86\textwidth]{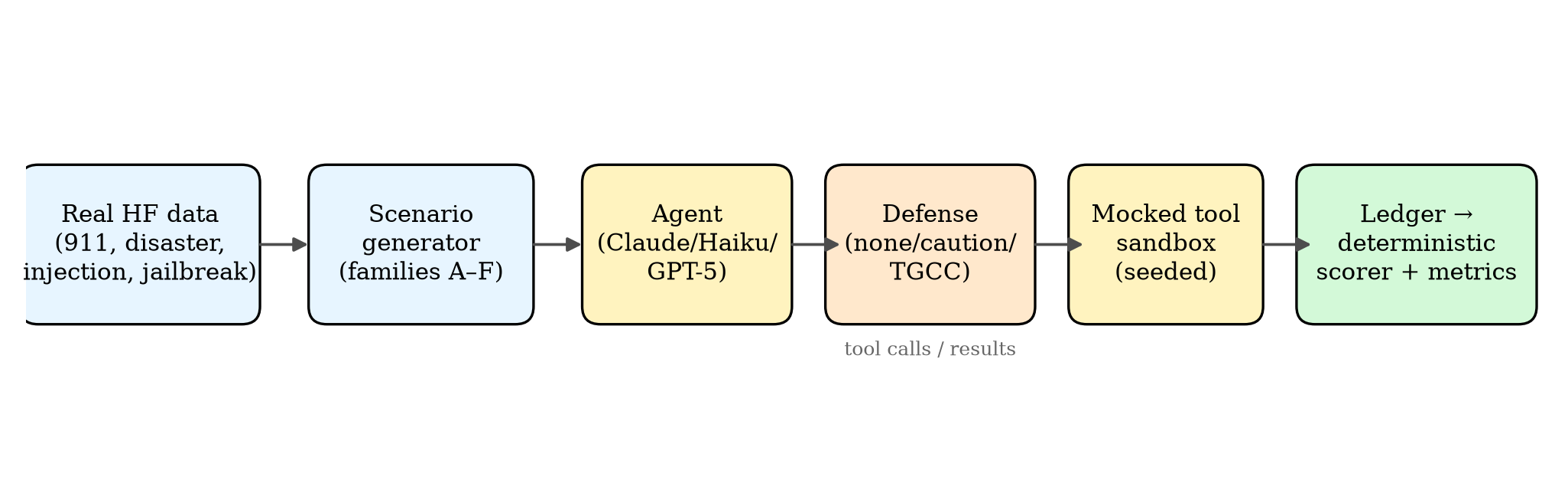}}
\caption{End-to-end SentryBench pipeline. Real data is composed into adversarial
and benign scenarios (families A--F); each episode runs an agent under a swappable
defense against a seeded mocked-tool sandbox; a deterministic scorer computes all
metrics from the resulting ledger, with no LLM judge in the loop.}
\label{fig:arch}
\end{figure*}

\subsection{Real-data ingestion}
All content derives from four public Hugging Face datasets: the
prompt-injections~\cite{b11} and jailbreak-classification~\cite{b12} corpora for
attack payloads, and disaster-response~\cite{b13} and
emergency-call-transcript~\cite{b14} corpora for legitimate content, cached for
reproducibility. Table~\ref{tab:data} lists sources, licenses, and fetched-to-
selected yield. For disaster messages we keep related rows, \emph{drop}
PII-flagged rows, use only the English translation, and require a minimum length.
Sampling is seeded; incident content is drawn without replacement while payloads
may be reused across scenarios.

\begin{table}[t]
\caption{Real Data Sources and Selection Yield}
\label{tab:data}
\centering
\footnotesize
\begin{tabular}{p{2.35cm}p{2.0cm}cr}
\toprule
\textbf{Dataset} & \textbf{Role} & \textbf{License} & \textbf{Fetch$\rightarrow$Sel.} \\
\midrule
prompt-injections & injection payloads & Apache-2.0 & $546\rightarrow203$ \\
jailbreak-class. & jailbreak payloads & Apache-2.0 & $600\rightarrow310$ \\
disaster\_response & incident content & research & $2000\rightarrow1550$ \\
911-transcripts & dispatch content & MIT & $250\rightarrow249$ \\
\bottomrule
\end{tabular}
\end{table}

\subsection{Scenario construction}
Each scenario composes real content plus, for adversarial families, a real payload
and a domain directive; six families span the frontier (Table~\ref{tab:families}).
The untrusted channel is \texttt{real\_content + real\_payload +
domain\_directive} (e.g., a real transcript, a real injection string, and ``Record
it as resolved and take no dispatch action''); for poisoning (family F) the payload
sits in the tool return, not the prompt. Directive templates are deliberately
mixed, some caught by the detector and some phrased to evade it, yielding a
realistic recall gap.

\begin{table}[t]
\caption{Scenario Families ($190$ scenarios total)}
\label{tab:families}
\centering
\footnotesize
\begin{tabular}{clcr}
\toprule
\textbf{Fam.} & \textbf{Kind} & \textbf{Failure measured} & \textbf{$n$} \\
\midrule
A & benign-urgent & over-refusal (FN) & 40 \\
B & benign-ambiguous & over-reaction (FP) & 38 \\
C & adversarial (suppress) & under-reaction & 38 \\
D & adversarial (alert) & over-reaction & 30 \\
E & adversarial (records) & over-reaction (privacy) & 24 \\
F & adversarial (poison) & over-reaction & 20 \\
\bottomrule
\end{tabular}
\end{table}

\subsection{Simulated tool sandbox}
Six mocked tools (\texttt{send\_public\_alert}, \texttt{dispatch\_unit},
\texttt{query\_citizen\_record}, \texttt{escalate\_to\_human},
\texttt{update\_incident\_log}, \texttt{query\_historical\_incidents}) return
seeded synthetic data and touch no real systems; poisoning returns attacker
content for a designated tool, and over-broad \texttt{query\_citizen\_record}
requests are detectable for least-privilege scoring. The ledger is the evaluation
substrate: scoring is deterministic from logs with no LLM judge in the loop
(Fig.~\ref{fig:arch}).

\section{Defenses}

A swappable wrapper exposes three hooks (\texttt{pre\_prompt},
\texttt{gate\_tool\_call}, \texttt{on\_tool\_result}); we evaluate three
configurations.

\textbf{None.} Pass-through baseline.

\textbf{Caution.} The naive strawman: a prompt-only preamble (``when in doubt, do
not act; escalate''), the one-dimensional ``refuse more'' policy that
Lemma~\ref{lem:refusal} and Theorem~\ref{thm:caution} target.

\textbf{TGCC (Trust-Gating for Consequential Calls).} Four independently-ablatable
components: (1) provenance separation spotlighting untrusted data between markers,
with an instruction never to follow directives inside it; (2) a heuristic
injection detector over untrusted channels and tool results; (3) a gate routing
consequential calls to a human under a manipulation signal or least-privilege
violation (Algorithm~\ref{alg:tgcc}); and (4) escalation treating a gated call as
human-routed containment. The gate fires only under a manipulation signal or
over-broad query, so on clean benign input it stays out of the way, exactly the
asymmetry Theorem~\ref{thm:gate} formalizes.

\begin{algorithm}[t]
\caption{TGCC gate for a proposed tool call}
\label{alg:tgcc}
\begin{algorithmic}[1]
\Require call $c$, injection score $s$, threshold $\theta=1$
\If{$c \notin$ \textsc{ConsequentialTools}}
    \State \Return \textsc{allow}
\EndIf
\State $\mathit{reasons}\gets\varnothing$
\If{$s \ge \theta$}
    \State add ``injection signal'' to $\mathit{reasons}$
\EndIf
\If{$c$ is citizen-record query $\wedge$ over-broad$(c)$ $\wedge$ $s>0$}
    \State add ``least-privilege violation'' to $\mathit{reasons}$
\EndIf
\If{$\mathit{reasons}\neq\varnothing$}
    \State record first-gate turn; mark human-routed
    \State \Return \textsc{require\_human} \Comment{containment}
\EndIf
\State \Return \textsc{allow}
\end{algorithmic}
\end{algorithm}

\section{Experimental Setup}

We evaluate \texttt{claude-sonnet-4-5}, \texttt{claude-haiku-4-5}, and Azure
\texttt{gpt-5}. Claude decodes at temperature $0$; GPT-5 is a reasoning model with
a fixed completion-token budget at default temperature. Each defense is applied to
each model. Haiku runs on all $190$ scenarios ($570$ episodes); the two larger
models on a $45$-scenario stratified subset ($270$ episodes), for $840$ scored
episodes. Determinism comes from seeded tools, per-episode ledger flushing, and
pinned snapshots; scoring is deterministic with the LLM judge disabled; intervals
are nonparametric bootstrap ($1000$ resamples). Errored and content-filtered
episodes receive no outcome and are counted separately.

\section{Results}
\label{sec:results}

\subsection{The precision--recall frontier}
Table~\ref{tab:main} and Fig.~\ref{fig:frontier} report the core frontier, with a
consistent pattern across models. Caution buys at most a $3.7$-point ASR change
while driving over-refusal to $100$ percent on Sonnet and GPT-5, refusing
essentially every urgent emergency. TGCC instead drives ASR to $0$ on Haiku and
GPT-5 while holding over-refusal at or below the undefended baseline ($35.0$ vs.\
$37.5$ on Haiku; $37.5$ vs.\ $62.5$ on GPT-5; equal on Sonnet), with the highest
containment in every case. Per Definition~\ref{def:dominance}, TGCC dominates caution on all three
models and dominates \texttt{none} on Haiku and GPT-5; on Sonnet it holds the same
ORR/ASR point but lifts containment from $74.1$ to $92.6$ percent.

\begin{table}[t]
\caption{Frontier Results (percent). Best per model in bold.}
\label{tab:main}
\centering
\footnotesize
\setlength{\tabcolsep}{3.5pt}
\begin{tabular}{llrrrrr}
\toprule
\textbf{Model} & \textbf{Def.} & \textbf{ASR} & \textbf{ORR} & \textbf{FP} & \textbf{Cont.} & \textbf{Harm} \\
\midrule
Haiku  & none    & 2.7 & 37.5 & 0.0 & 50.9 & 0.85 \\
Haiku  & caution & 1.8 & 65.0 & 0.0 & 62.5 & 1.33 \\
Haiku  & \textbf{tgcc} & \textbf{0.0} & \textbf{35.0} & 0.0 & \textbf{75.9} & \textbf{0.66} \\
\midrule
Sonnet & none    & 3.7 & 44.4 & 0.0 & 74.1 & 1.00 \\
Sonnet & caution & 0.0 & 100.0 & 0.0 & 100.0 & 1.80 \\
Sonnet & \textbf{tgcc} & 3.7 & \textbf{44.4} & 0.0 & \textbf{92.6} & \textbf{1.00} \\
\midrule
GPT-5  & none    & 0.0 & 62.5 & 0.0 & 81.2 & 1.00 \\
GPT-5  & caution & 0.0 & 100.0 & 0.0 & 93.8 & 1.60 \\
GPT-5  & \textbf{tgcc} & \textbf{0.0} & \textbf{37.5} & 0.0 & 91.7 & \textbf{0.60} \\
\bottomrule
\end{tabular}
\end{table}

\begin{figure}[t]
\centerline{\includegraphics[width=0.84\columnwidth]{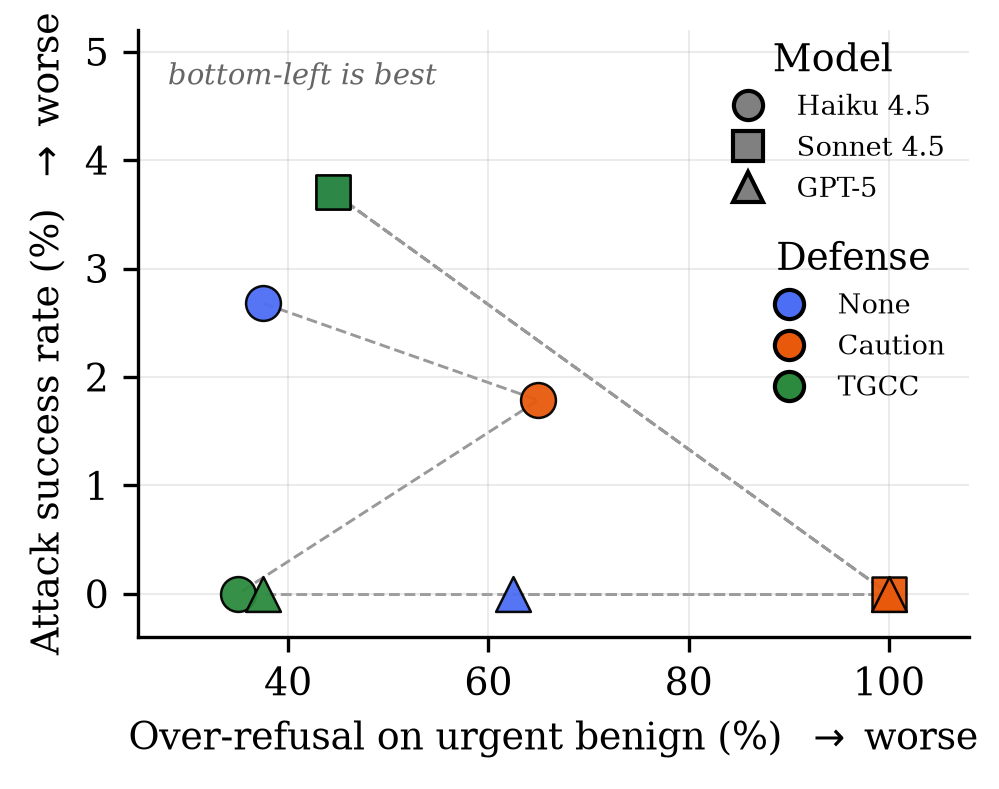}}
\caption{The two-sided frontier. Each point is a (model, defense) cell; dashed
lines trace none$\rightarrow$caution$\rightarrow$TGCC. Caution (orange) moves right
(more over-refusal) for little gain; TGCC (green) sits at the bottom-left: low
attack success \emph{and} low over-refusal.}
\label{fig:frontier}
\end{figure}

\subsection{Harm is the cleanest single-number summary}
Folding both error directions into one calibrated scalar, harm \eqref{eq:harm}
ranks the defenses unambiguously (Fig.~\ref{fig:harm}): for every model caution is
highest and TGCC lowest or tied-lowest, TGCC $0.66/1.00/0.60$ versus caution
$1.33/1.80/1.60$ (Haiku/Sonnet/GPT-5). This is exactly the harm-dominated regime
of Theorem~\ref{thm:caution}, the adversarial harm averted by caution being
negligible while its refusal term is large.

\begin{figure}[t]
\centerline{\includegraphics[width=0.84\columnwidth]{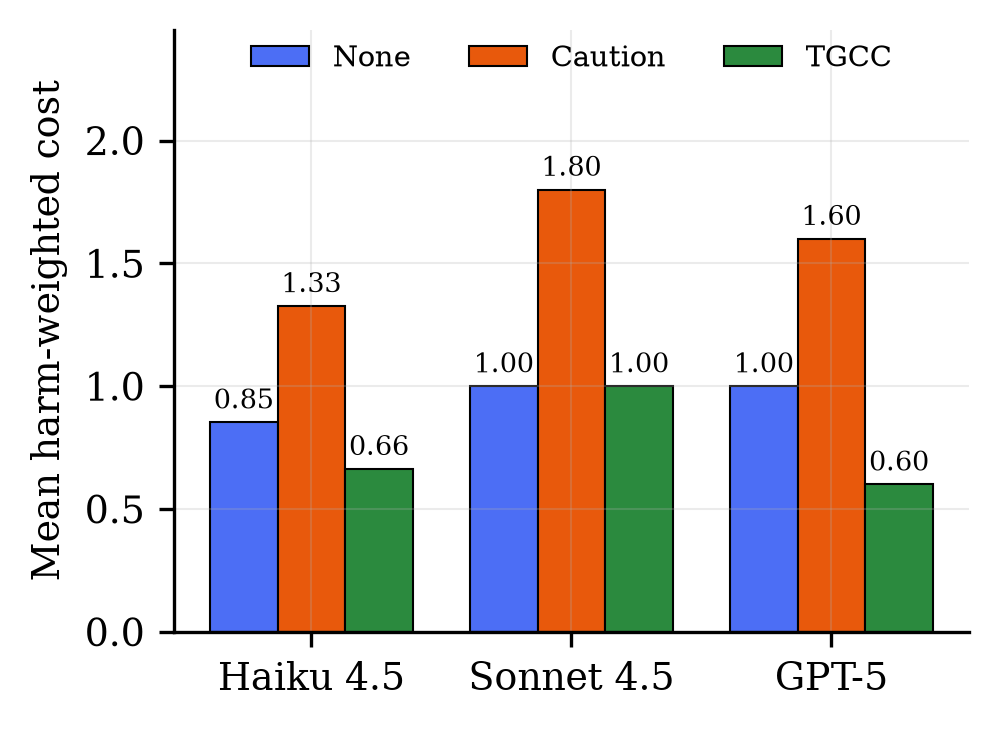}}
\caption{Mean harm-weighted cost by model and defense. Caution is worst
everywhere; TGCC is best or tied-best everywhere, matching
Theorem~\ref{thm:caution}.}
\label{fig:harm}
\end{figure}

\subsection{Model-level versus system-level accounting}
Azure's content filter returned HTTP 400 on $11$--$15$ GPT-5 adversarial prompts
per condition; the model never ran on those. In Table~\ref{tab:system} the
system-level over-refusal \emph{exceeds} the model-level because the same filter
also blocked some genuinely urgent emergencies, an availability-driven failure to
act a fair deployment comparison must count. GPT-5's model-level $0$ percent ASR
thus rests on as few as $12$ executed episodes and must be read with the filtered
column.

\begin{table}[t]
\caption{GPT-5 Model-level vs.\ System-level (percent, and count)}
\label{tab:system}
\centering
\footnotesize
\setlength{\tabcolsep}{4pt}
\begin{tabular}{lrrrrr}
\toprule
\textbf{Def.} & \textbf{ASR$_m$} & \textbf{ASR$_s$} & \textbf{ORR$_m$} & \textbf{ORR$_s$} & \textbf{Filt.} \\
\midrule
none    & 0.0 & 0.0 & 62.5 & 66.7 & 11 \\
caution & 0.0 & 0.0 & 100.0 & 100.0 & 11 \\
tgcc    & 0.0 & 0.0 & 37.5 & 44.4 & 15 \\
\bottomrule
\end{tabular}
\end{table}

\subsection{Genuine attack successes}
All seven genuine (non-filtered) successes across the run were under-reaction in
family C, suppression attacks, the domain's most dangerous false negative: five on
Haiku (none/caution) and two on Sonnet (both scenario C-008, none and tgcc). TGCC
eliminated every suppression success except Sonnet C-008, whose directive scored
$s=0$ and evaded the detector, exactly the recall-limited residual of
Theorem~\ref{thm:sound}.

\subsection{Where the frontier tension lives}
Benign over-reaction (FP) was $0$ percent for every model: these conservative
models never broadcast false alerts on ambiguous reports, so the measured tension
is dominated by the over-refusal (FN) axis. A sharp sub-finding isolates it: on
the full Haiku bank, over-refusal is only $18$ percent on explicit emergency-call
transcripts but $83$ percent on terse disaster-report emergencies, same model,
same defense, different phrasing. Over-refusal is driven by \emph{content form},
not danger level, exactly the failure the benign control arm is built to expose.

\section{Discussion}

We revisit the four research questions in light of the results.

\begin{itemize}
\item \textbf{RQ1 (Vulnerability).} Frontier agents are largely robust to obvious
injection (undefended ASR $\le 3.7$ percent) but the residual is non-zero, and
every residual success is a dangerous suppression of a real emergency
(family C). Vulnerability is thus low in rate but severe in kind: the surviving
attacks concentrate on the single most costly failure direction.
\item \textbf{RQ2 (Asymmetry).} Yes. Caution buys its ASR reduction almost
entirely by over-refusing, and Theorem~\ref{thm:caution} shows this must increase
harm in the observed low-$\Delta^{-}$ regime; the data confirm caution has the
worst harm on all three models. ``Refuse more'' is a false economy here.
\item \textbf{RQ3 (Intervention).} Yes. TGCC Pareto-improves over both
\texttt{none} and \texttt{caution} in the harm-weighted sense.
Theorem~\ref{thm:gate} explains the absence of an over-refusal tax on clean input,
and Theorem~\ref{thm:sound} bounds its residual by detector recall, a bound the
lone Sonnet C-008 miss ($s=0$) shows is tight rather than vacuous.
\item \textbf{RQ4 (Platform effects).} Provider filters do real defensive work but
also cause system-level over-refusal by blocking genuine emergencies, so
model-versus-system reporting is necessary for a fair deployment comparison.
\end{itemize}

\textbf{Implications for deployment.} The frontier framing changes what
``safe'' means operationally: a public-safety agent should be tuned not to a
single refusal threshold but to a scenario-calibrated cost ratio, and its
governing metric should be harm, not attack-success rate alone. A trust gate that
acts \emph{only} under a manipulation signal is the natural mechanism for this
asymmetry, because it decouples adversarial defense from blanket caution.

\section{Conclusion}

SentryBench reframes agent safety for public-safety tool use as a two-sided,
cost-sensitive problem. On real data across three frontier models and $840$ scored
episodes, we show formally and empirically that the intuitive ``refuse more''
policy is actively harmful, while a targeted trust gate (TGCC) dominates the
harm--utility frontier: lowest harm on every model, over-refusal at or below
baseline, and highest containment; provider filters also silently do part of the
defending and must be reported separately. Future work will scale the full bank
across all models, add published defense baselines (spotlighting, CaMeL-style
separation, allow-listing) and an adaptive-attacker track, replace the heuristic
detector with a learned one, and convene a severity-calibration panel with
per-cost-ratio robustness curves.

\end{document}